\documentclass[12pt]{article}%
\usepackage{amsmath}
\usepackage{mitpress}
\usepackage{amsfonts}
\usepackage{amssymb}
\usepackage{graphicx}%
\providecommand{\U}[1]{\protect\rule{.1in}{.1in}}
\newtheorem{theorem}{Theorem}

\newtheorem{definition}[theorem]{Definition}
\newtheorem{example}[theorem]{Example}

\newtheorem{notation}[theorem]{Notation}

\newtheorem{remark}[theorem]{Remark}

\newenvironment{proof}[1][Proof]{\textbf{#1.} }{\ \rule{0.5em}{0.5em}}
\newdimen\dummy
\dummy=\oddsidemargin
\begin{document}

\title{Universal Regular Autonomous Asynchronous Systems: Fixed Points, Equivalencies
and Dynamic Bifurcations}
\author{Serban E. Vlad\\str. Zimbrului, nr. 3, bl. PB68, ap. 11, 410430, Oradea, Romania, E-mail: serban\_e\_vlad@yahoo.com}
\maketitle

\begin{abstract}
The asynchronous systems are the non-deterministic models of the asynchronous
circuits from the digital electrical engineering. In the autonomous version,
such a system is a set of functions $x:\mathbf{R}\rightarrow\{0,1\}^{n}$
called states ($\mathbf{R}$ is the time set). If an asynchronous system is
defined by making use of a so called generator function $\Phi:\{0,1\}^{n}%
\rightarrow\{0,1\}^{n},$ then it is called regular. The property of
universality means the greatest in the sense of the inclusion.

The purpose of the paper is that of defining and of characterizing the fixed
points, the equivalencies and the dynamical bifurcations of the universal
regular autonomous asynchronous systems. We use analogies with the dynamical
systems theory.

\end{abstract}

\section{Preliminaries}

\begin{definition}
We denote by $\mathbf{B}=\{0,1\}$ the \textbf{binary Boole algebra}, endowed
with the discrete topology and with the usual laws.
\end{definition}

\begin{definition}
Let be the Boolean function $\Phi:\mathbf{B}^{n}\rightarrow\mathbf{B}^{n}%
,\Phi=(\Phi_{1},...,\Phi_{n})$ and $\nu\in\mathbf{B}^{n},\nu=(\nu_{1}%
,...,\nu_{n}).$ We define $\Phi^{\nu}:\mathbf{B}^{n}\rightarrow\mathbf{B}^{n}$
by $\forall\mu\in\mathbf{B}^{n},$%
\[
\Phi^{\nu}(\mu)=(\overline{\nu_{1}}\cdot\mu_{1}\oplus\nu_{1}\cdot\Phi_{1}%
(\mu),...,\overline{\nu_{n}}\cdot\mu_{n}\oplus\nu_{n}\cdot\Phi_{n}(\mu)).
\]

\end{definition}

\begin{remark}
$\Phi^{\nu}$ represents the function resulting from $\Phi$ when this one is
not computed, in general, on all the coordinates $\Phi_{i},i=\overline{1,n}:$
if $\nu_{i}=0,$ then $\Phi_{i}$ is not computed, $\Phi_{i}^{\nu}(\mu)=\mu_{i}$
and if $\nu_{i}=1,$ then $\Phi_{i}$ is computed, $\Phi_{i}^{\nu}(\mu)=\Phi
_{i}(\mu).$
\end{remark}

\begin{definition}
Let be the sequence $\alpha^{0},\alpha^{1},...,\alpha^{k},...\in\mathbf{B}%
^{n}.$ The functions $\Phi^{\alpha^{0}\alpha^{1}...\alpha^{k}}:\mathbf{B}%
^{n}\rightarrow\mathbf{B}^{n}$ are defined iteratively by $\forall
k\in\mathbf{N},\forall\mu\in\mathbf{B}^{n},$%
\[
\Phi^{\alpha^{0}\alpha^{1}...\alpha^{k}\alpha^{k+1}}(\mu)=\Phi^{\alpha^{k+1}%
}(\Phi^{\alpha^{0}\alpha^{1}...\alpha^{k}}(\mu)).
\]

\end{definition}

\begin{definition}
The sequence $\alpha^{0},\alpha^{1},...,\alpha^{k},...\in\mathbf{B}^{n}$ is
called \textbf{progressive} if%
\[
\forall i\in\{1,...,n\},\text{ the set }\{k|k\in\mathbf{N},\alpha_{i}%
^{k}=1\}\text{ is infinite.}%
\]
The set of the progressive sequences is denoted by $\Pi_{n}.$
\end{definition}

\begin{remark}
Let be $\mu\in\mathbf{B}^{n}.$ When $\alpha=\alpha^{0},\alpha^{1}%
,...,\alpha^{k},...$ is progressive, each coordinate $\Phi_{i},i=\overline
{1,n}$ is computed infinitely many times in the sequence $\Phi^{\alpha
^{0}\alpha^{1}...\alpha^{k}}(\mu),$ $k\in\mathbf{N}$. This is the meaning of
the progress property, giving the so called 'unbounded delay model' of
computation of the Boolean functions.
\end{remark}

\begin{definition}
The \textbf{initial value}, denoted by $x(-\infty+0)$ or $\underset
{t\rightarrow-\infty}{\lim}x(t)\in\mathbf{B}^{n}$ and the \textbf{final
value}, denoted by $x(\infty-0)$ or $\underset{t\rightarrow\infty}{\lim
}x(t)\in\mathbf{B}^{n}$ of the function $x:\mathbf{R}\rightarrow\mathbf{B}%
^{n}$ are defined by%
\[
\exists t^{\prime}\in\mathbf{R},\forall t<t^{\prime},x(t)=x(-\infty+0),
\]%
\[
\exists t^{\prime}\in\mathbf{R},\forall t>t^{\prime},x(t)=x(\infty-0).
\]

\end{definition}

\begin{definition}
The function $x:\mathbf{R}\rightarrow\mathbf{B}^{n}$ is called
(\textbf{pseudo})\textbf{periodical with the period} $T_{0}>0$ if

a) $\underset{t\rightarrow\infty}{\lim}x(t)$ does not exist and

b) $\exists t^{\prime}\in\mathbf{R},\forall t\geq t^{\prime},x(t)=x(t+T_{0}).$
\end{definition}

\begin{definition}
The \textbf{characteristic function} $\chi_{A}:\mathbf{R}\rightarrow
\mathbf{B}$ of the set $A\subset\mathbf{R}$ is defined in the following way:%
\[
\chi_{A}(t)=\left\{
\begin{array}
[c]{c}%
1,\;if\;t\in A\;\\
0,otherwise
\end{array}
\right.  .
\]

\end{definition}

\begin{notation}
We denote by $Seq$ the set of the real sequences $t_{0}<t_{1}<...<t_{k}<...$
which are unbounded from above.
\end{notation}

\begin{remark}
The sequences $(t_{k})\in Seq$ act as time sets. At this level of generality
of the exposure, a double uncertainty exists in the real time iterative
computations of the function $\Phi:\mathbf{B}^{n}\rightarrow\mathbf{B}^{n}: $
we do not know precisely neither the coordinates $\Phi_{i}$ of $\Phi$ that are
computed, nor when the computation happens. This uncertainty implies the non
determinism of the model and its origin consists in structural fluctuations in
the fabrication process, the variations in ambiental temperature and the power
supply etc.
\end{remark}

\begin{definition}
A \textbf{signal} (or $n-$\textbf{signal}) is a function $x:\mathbf{R}%
\rightarrow\mathbf{B}^{n}$ of the form%
\begin{equation}
x(t)=x(-\infty+0)\cdot\chi_{(-\infty,t_{0})}(t)\oplus x(t_{0})\cdot
\chi_{\lbrack t_{0},t_{1})}(t)\oplus... \label{pre1}%
\end{equation}%
\[
...\oplus x(t_{k})\cdot\chi_{\lbrack t_{k},t_{k+1})}(t)\oplus...
\]
with $(t_{k})\in Seq.$ The set of the signals is denoted by $S^{(n)}.$
\end{definition}

\begin{remark}
The signals $x\in S^{(n)}$ model the electrical signals from the digital
electrical engineering. They have by definition initial values and they avoid
'Dirichlet type' properties (called Zeno properties by the engineers) such as%
\[
\exists t\in\mathbf{R},\forall\varepsilon>0,\exists t^{\prime}\in
(t-\varepsilon,t),\exists t^{\prime\prime}\in(t-\varepsilon,t),x(t^{\prime
})\neq x(t^{\prime\prime}),
\]%
\[
\exists t\in\mathbf{R},\forall\varepsilon>0,\exists t^{\prime}\in
(t,t+\varepsilon),\exists t^{\prime\prime}\in(t,t+\varepsilon),x(t^{\prime
})\neq x(t^{\prime\prime})
\]
because these properties cannot characterize the inertial devices.
\end{remark}

\begin{notation}
We denote by $P^{\ast}$ the set of the non-empty subsets of a set.
\end{notation}

\begin{definition}
The \textbf{autonomous asynchronous systems} are the non-empty sets $X\in
P^{\ast}(S^{(n)}).$
\end{definition}

\begin{example}
\label{Exa13}We give the following simple example that shows how the
autonomous asynchronous systems model the asynchronous circuits. In Figure
\ref{preliminaries1} we have drawn the (logical) gate NOT with the%
\begin{figure}
[ptb]
\begin{center}
\fbox{\includegraphics[
height=0.7308in,
width=1.4451in
]%
{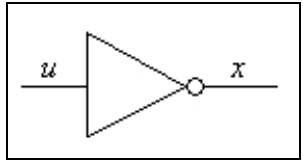}%
}\caption{Circuit with the logical gate NOT}%
\label{preliminaries1}%
\end{center}
\end{figure}
input $u\in S^{(1)}$ and the state (the output) $x\in S^{(1)}.$ For
$\lambda\in\mathbf{B}$ and%
\[
u(t)=\lambda,
\]
the state $x$ represents the computation of the negation of $u$ and it is of
the form%
\[
x(t)=\mu\cdot\chi_{(-\infty,t_{0})}(t)\oplus\overline{\lambda}\cdot
\chi_{\lbrack t_{0},t_{1})}(t)\oplus\overline{\lambda}\cdot\chi_{\lbrack
t_{1},t_{2})}(t)\oplus...\oplus\overline{\lambda}\cdot\chi_{\lbrack
t_{k},t_{k+1})}(t)\oplus...
\]%
\[
=\mu\cdot\chi_{(-\infty,t_{0})}(t)\oplus\overline{\lambda}\cdot\chi_{\lbrack
t_{0},\infty)}(t),
\]
where $\mu\in\mathbf{B}$ is the initial value of $x$ and $(t_{k})\in Seq$ is
arbitrary. As we can see, $x$ depends on $t_{0},\mu,\lambda$ only and it is
independent on $t_{1},t_{2},...$

In Figure \ref{preliminaries2},%
\begin{figure}
[ptb]
\begin{center}
\fbox{\includegraphics[
height=0.9548in,
width=1.5031in
]%
{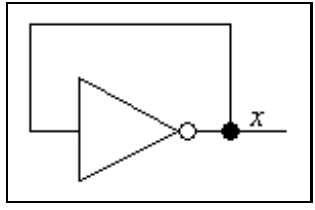}%
}\caption{Circuit with feedback with the logical gate NOT}%
\label{preliminaries2}%
\end{center}
\end{figure}
we have%
\[
x(t)=\mu\cdot\chi_{(-\infty,t_{0})}(t)\oplus\overline{\mu}\cdot\chi_{\lbrack
t_{0},t_{1})}(t)\oplus\mu\cdot\chi_{\lbrack t_{1},t_{2})}(t)\oplus...
\]%
\[
\oplus\overline{\mu}\cdot\chi_{\lbrack t_{2k},t_{2k+1})}(t)\oplus\mu\cdot
\chi_{\lbrack t_{2k+1},t_{2k+2})}(t)\oplus...
\]
thus this circuit is modeled by the autonomous asynchronous system%
\[
X=\{\mu\cdot\chi_{(-\infty,t_{0})}(t)\oplus\overline{\mu}\cdot\chi_{\lbrack
t_{0},t_{1})}(t)\oplus\mu\cdot\chi_{\lbrack t_{1},t_{2})}(t)\oplus...
\]%
\[
\oplus\overline{\mu}\cdot\chi_{\lbrack t_{2k},t_{2k+1})}(t)\oplus\mu\cdot
\chi_{\lbrack t_{2k+1},t_{2k+2})}(t)\oplus...|\mu\in\mathbf{B},(t_{k})\in
Seq\}\in P^{\ast}(S^{(1)}).
\]

\end{example}

\begin{definition}
The \textbf{progressive functions} $\rho:\mathbf{R}\rightarrow\mathbf{B}^{n}$
are by definition the functions%
\begin{equation}
\rho(t)=\alpha^{0}\cdot\chi_{\{t_{0}\}}(t)\oplus\alpha^{1}\cdot\chi
_{\{t_{1}\}}(t)\oplus...\oplus\alpha^{k}\cdot\chi_{\{t_{k}\}}(t)\oplus...
\label{pre2}%
\end{equation}
where $(t_{k})\in Seq$ and $\alpha^{0},\alpha^{1},...,\alpha^{k},...\in\Pi
_{n}.$ The set of the progressive functions is denoted by $P_{n}.$
\end{definition}

\begin{definition}
\label{Def20}For $\Phi:\mathbf{B}^{n}\rightarrow\mathbf{B}^{n}$ and $\rho\in
P_{n}$ like at (\ref{pre2}), we define $\Phi^{\rho}:\mathbf{R}\times
\mathbf{B}^{n}\rightarrow\mathbf{B}^{n}$ by $\forall t\in\mathbf{R},\forall
\mu\in\mathbf{B}^{n},$%
\[
\Phi^{\rho}(t,\mu)=\mu\cdot\chi_{(-\infty,t_{0})}(t)\oplus\Phi^{\alpha^{0}%
}(\mu)\cdot\chi_{\lbrack t_{0},t_{1})}(t)\oplus...\oplus\Phi^{\alpha
^{0}...\alpha^{k}}(\mu)\cdot\chi_{\lbrack t_{k},t_{k+1})}(t)\oplus...
\]

\end{definition}

\begin{remark}
The previous equation reminds the iterations of a discrete time real dynamical
system. The time is not exactly discrete in it, but some sort of intermediate
situation occurs between the discrete and the real time; on the other hand the
iterations of $\Phi$ do not happen on all the coordinates (synchronicity), but
on some coordinates only, such that any coordinate $\Phi_{i}$ is computed
infinitely many times, $i=\overline{1,n}$ (asynchronicity) when $t\in
\mathbf{R}$.
\end{remark}

\section{Discrete time}

\begin{notation}
We denote by%
\[
\mathbf{N}_{\_}=\mathbf{N}\cup\{-1\}
\]
the discrete time set.
\end{notation}

\begin{definition}
Let be $\Phi:\mathbf{B}^{n}\rightarrow\mathbf{B}^{n}$ and $\alpha\in\Pi
_{n},\alpha=\alpha^{0},...,\alpha^{k},...$ We define the function
$\widehat{\Phi}^{\alpha}:\mathbf{N}_{\_}\times\mathbf{B}^{n}\rightarrow
\mathbf{B}^{n}$ by $\forall(k,\mu)\in\mathbf{N}_{\_}\times\mathbf{B}^{n}, $%
\[
\widehat{\Phi}^{\alpha}(k,\mu)=\left\{
\begin{array}
[c]{c}%
\mu,k=-1,\\
\Phi^{\alpha^{0}...\alpha^{k}}(\mu),k\geq0
\end{array}
\right.  .
\]

\end{definition}

\begin{notation}
Let us denote%
\[
\widehat{\Pi}_{n}=\{\alpha|\alpha\in\Pi_{n},\forall k\in\mathbf{N},\alpha
^{k}\neq(0,...,0)\}.
\]

\end{notation}

\begin{definition}
The equivalence of $\rho,\rho^{\prime}\in P_{n}$ is defined by: $\exists
(t_{k})\in Seq,\exists(t_{k}^{\prime})\in Seq,\exists\alpha\in\widehat{\Pi
}_{n}$ such that (\ref{pre2}) and%
\[
\rho^{\prime}(t)=\alpha^{0}\cdot\chi_{\{t_{0}^{\prime}\}}(t)\oplus\alpha
^{1}\cdot\chi_{\{t_{1}^{\prime}\}}(t)\oplus...\oplus\alpha^{k}\cdot
\chi_{\{t_{k}^{\prime}\}}(t)\oplus...
\]
are true.
\end{definition}

\begin{definition}
The 'canonical surjection' $s:P_{n}\rightarrow\widehat{\Pi}_{n}$ is by
definition the function $\forall\rho\in P_{n},$%
\[
s(\rho)=\alpha
\]
where $\alpha\in\widehat{\Pi}_{n}$ is the only sequence such that $(t_{k})\in
Seq$ exists, making the equation (\ref{pre2}) true.
\end{definition}

\begin{remark}
The relation between the continuous and the discrete time is the following:
for any $\mu\in\mathbf{B}^{n}$ and any $\rho\in P_{n},$ $\alpha\in\widehat
{\Pi}_{n}$ and $(t_{k})\in Seq$ exist making the equation (\ref{pre2}) true
and we have%
\[
\Phi^{\rho}(t,\mu)=\widehat{\Phi}^{\alpha}(-1,\mu)\cdot\chi_{(-\infty,t_{0}%
)}(t)\oplus\widehat{\Phi}^{\alpha}(0,\mu)\cdot\chi_{\lbrack t_{0},t_{1}%
)}(t)\oplus...
\]%
\[
...\oplus\widehat{\Phi}^{\alpha}(k,\mu)\cdot\chi_{\lbrack t_{k},t_{k+1}%
)}(t)\oplus...
\]
Equivalent progressive functions $\rho,\rho^{\prime}\in P_{n}$ (i.e.
$s(\rho)=s(\rho^{\prime})$) give 'equivalent' functions $\Phi^{\rho}%
(t,\mu),\Phi^{\rho^{\prime}}(t,\mu)$ in the sense that the computations of
$\Phi$ are the same, but the time flow is piecewise faster or slower in the
two situations.
\end{remark}

\section{Regular autonomous asynchronous systems}

\begin{definition}
The \textbf{universal regular autonomous asynchronous system} $\Xi_{\Phi}\in
P^{\ast}(S^{(n)})$ that is generated by the function $\Phi:\mathbf{B}%
^{n}\rightarrow\mathbf{B}^{n}$ is defined by
\[
\Xi_{\Phi}=\{\Phi^{\rho}(\cdot,\mu)|\mu\in\mathbf{B}^{n},\rho\in P_{n}\}.
\]

\end{definition}

\begin{definition}
An autonomous asynchronous system $X\in P^{\ast}(S^{(n)})$ is called
\textbf{regular}, if $\Phi$ exists such that $X\subset\Xi_{\Phi}.$ In this
case $\Phi$ is called the \textbf{generator function\footnote{The terminology
of 'generator function' is also used in \cite{bib4}, page 18 meaning the
vector field of a discrete time dynamical system. In \cite{bib6} the
terminology of 'generator' (function) of a dynamical system is mentioned too.
Moisil called $\Phi$ 'network function' in a non-autonomous, discrete time
context; for Moisil, 'network' means 'system' or 'circuit'.}} of $X$.
\end{definition}

\begin{remark}
In the last two definitions, the attribute 'regular' refers to the existence
of a generator function $\Phi$ and the attribute 'universal' means maximal
relative to the inclusion.

For a regular system, $\Phi$ is not unique in general.
\end{remark}

\begin{example}
For any $\mu^{0}\in\mathbf{B}^{n}$ and $\rho^{\ast}\in P_{n},$ the autonomous
systems $\{\Phi^{\rho^{\ast}}(\cdot,\mu^{0})\},$ $\{\Phi^{\rho}(\cdot,\mu
^{0})|\rho\in P_{n}\},$ $\{\Phi^{\rho^{\ast}}(\cdot,\mu)|\mu\in\mathbf{B}%
^{n}\}$ and $\Xi_{\Phi}$ are regular.

For $\Phi=1_{\mathbf{B}^{n}},$ the system $\Xi_{1_{\mathbf{B}^{n}}}=\{\mu
|\mu\in\mathbf{B}^{n}\}=\mathbf{B}^{n}$ is regular.

Another example of universal regular autonomous asynchronous system is given
by $\Phi=\mu^{0},$ the constant function, for which $\Xi_{\mu^{0}}%
=\{x|x_{i}=\mu_{i}\cdot\chi_{(-\infty,t_{i})}\oplus\mu_{i}^{0}\cdot
\chi_{\lbrack t_{i},\infty)},\mu_{i}\in\mathbf{B},t_{i}\in\mathbf{R}%
,i=\overline{1,n}\}$.
\end{example}

\begin{remark}
These examples suggest several possibilities of defining the systems
$X\subset\Xi_{\Phi}$ which are not universal. For example by putting
appropriate supplementary requests on the functions $\rho,$ one could
rediscover the 'bounded delay model' of computation of the Boolean functions.
\end{remark}

\section{Orbits and state portraits}

\begin{definition}
Let be $\rho\in P_{n}.$ Two things are understood by \textbf{orbit}, or
(\textbf{state}, or \textbf{phase}) \textbf{trajectory} (\cite{bib4}, page 19;
\cite{bib3}, page 3; \cite{bib1}, page 8; \cite{bib2}, page 24; \cite{bib5},
page 2) \textbf{of} $\Xi_{\Phi}$ \textbf{starting at} $\mu\in\mathbf{B}^{n}$:

a) the function $\Phi^{\rho}(\cdot,\mu):\mathbf{R}\rightarrow\mathbf{B}^{n};$

b) the set $Or_{\rho}(\mu)=\{\Phi^{\rho}(t,\mu)|t\in\mathbf{R}\}$ representing
the values of the previous function.

Sometimes (\cite{bib3}, page 4; \cite{bib6}, page 91; \cite{bib2}, page 24;
\cite{bib5}, page 2) the function from a) is called the \textbf{motion} (or
the \textbf{dynamic}) of $\mu$ through $\Phi^{\rho}.$
\end{definition}

\begin{definition}
The equivalent properties
\[
\exists t\in\mathbf{R},\Phi^{\rho}(t,\mu)=\mu^{\prime}%
\]
and
\[
\mu^{\prime}\in Or_{\rho}(\mu)
\]
are called of \textbf{accessibility}; the points $\mu^{\prime}\in Or_{\rho
}(\mu)$ are said to be \textbf{accessible}.
\end{definition}

\begin{remark}
\label{Rem32}The orbits are the curves in $\mathbf{B}^{n}$, parametrized by
$\rho$ and $t.$ On the other hand $\rho\in P_{n},$ $t^{\prime}\in\mathbf{R}$
imply $\rho\cdot\chi_{(t^{\prime},\infty)}\in P_{n}$ and we see the truth of
the implication%
\[
\mu^{\prime}=\Phi^{\rho}(t^{\prime},\mu)\Longrightarrow\forall t\geq
t^{\prime},\Phi^{\rho}(t,\mu)=\Phi^{\rho\cdot\chi_{(t^{\prime},\infty)}}%
(t,\mu^{\prime}).
\]

\end{remark}

\begin{definition}
The \textbf{state} (or the \textbf{phase})\textbf{\ portrait} of $\Xi_{\Phi}$
is the set of its orbits (\cite{bib3}, page 4; \cite{bib6}, page 92;
\cite{bib1}, page 10; \cite{bib5}, page 2).
\end{definition}

\begin{example}
\label{Exa34}The function $\Phi:\mathbf{B}^{2}\rightarrow\mathbf{B}^{2}$ is
defined by the following table%
\[%
\begin{array}
[c]{cc}%
(\mu_{1},\mu_{2}) & \Phi(\mu_{1},\mu_{2})\\
(0,0) & (0,0)\\
(0,1) & (1,0)\\
(1,0) & (1,1)\\
(1,1) & (1,1)
\end{array}
\]
The state portrait of $\Xi_{\Phi}$ is:%
\[
\{(0,1)\cdot\chi_{(-\infty,t_{0})}\oplus(0,0)\cdot\chi_{\lbrack t_{0},\infty
)}|t_{0}\in\mathbf{R}\}\cup
\]%
\[
\cup\{(0,1)\cdot\chi_{(-\infty,t_{0})}\oplus(1,0)\cdot\chi_{\lbrack
t_{0},t_{1})}\oplus(1,1)\cdot\chi_{\lbrack t_{1},\infty)}|t_{0},t_{1}%
\in\mathbf{R},t_{0}<t_{1}\}\cup
\]%
\[
\cup\{(0,1)\cdot\chi_{(-\infty,t_{0})}\oplus(1,1)\cdot\chi_{\lbrack
t_{0},\infty)}|t_{0}\in\mathbf{R}\}\cup
\]%
\[
\cup\{(1,0)\cdot\chi_{(-\infty,t_{0})}\oplus(1,1)\cdot\chi_{\lbrack
t_{0},\infty)}|t_{0}\in\mathbf{R}\}\cup\{(0,0)\}\cup\{(1,1)\}.
\]
This set is drawn in Figure \ref{ph1},%
\begin{figure}
[ptb]
\begin{center}
\fbox{\includegraphics[
height=1.1052in,
width=1.4131in
]%
{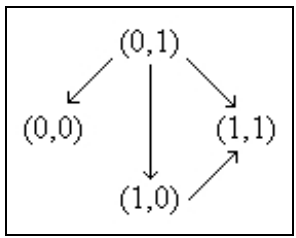}%
}\caption{The state portrait of the system from Example \ref{Exa34}.}%
\label{ph1}%
\end{center}
\end{figure}
where the arrows show the increase of time. One might want to put arrows from
$(0,0)$ to itself and from $(1,1)$ to itself.
\end{example}

\section{Nullclins}

\begin{definition}
Let be $\Phi:\mathbf{B}^{n}\rightarrow\mathbf{B}^{n}.$ For any $i\in
\{1,...,n\},$ the \textbf{nullclins} of $\Phi$ are the sets%
\[
NC_{i}=\{\mu|\mu\in\mathbf{B}^{n},\Phi_{i}(\mu)=\mu_{i}\}.
\]
If $\mu\in NC_{i},$ then the coordinate $i$ is said to be \textbf{not
excited}, or \textbf{not enabled}, or \textbf{stable} and if $\mu\in
\mathbf{B}^{n}\setminus NC_{i}$ then it is called \textbf{excited}, or
\textbf{enabled}, or \textbf{unstable}.
\end{definition}

\begin{remark}
Sometimes, instead of indicating $\Phi$ by a table like previously, we can
replace Figure \ref{ph1} by Figure \ref{ph2},%
\begin{figure}
[ptb]
\begin{center}
\fbox{\includegraphics[
height=1.0966in,
width=1.4096in
]%
{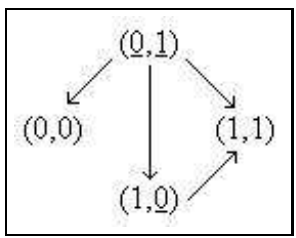}%
}\caption{The state portrait of the system from Example \ref{Exa34}, version}%
\label{ph2}%
\end{center}
\end{figure}
where we have underlined the unstable coordinates. For example in Figure
\ref{ph2}, $(\underline{0},\underline{1})$ means that $\Phi(0,1)=(1,0),$
$(1,\underline{0})$ means that $\Phi(1,0)=(1,1)$ etc.

In fact Figure \ref{ph2} results uniquely from Figure \ref{ph1}, one could
know by looking at Figure \ref{ph1} which coordinates should be underlined and
which should be not.
\end{remark}

\section{Fixed points}

\begin{definition}
A point $\mu\in\mathbf{B}^{n}$ that fulfills $\Phi(\mu)=\mu$ is called a
\textbf{fixed point} (an \textbf{equilibrium point}, a \textbf{critical
point}, a \textbf{singular point}) (\cite{bib4}, page 43; \cite{bib3}, page 4;
\cite{bib6}, page 92; \cite{bib1}, page 9; \cite{bib2}, page 24; \cite{bib5},
page 2), shortly an \textbf{equilibrium} of $\Phi.$ A point that is not fixed
is called \textbf{ordinary}.
\end{definition}

\begin{theorem}
\label{The38}The following statements are equivalent for $\mu\in\mathbf{B}%
^{n}:$%
\begin{equation}
\Phi(\mu)=\mu, \label{equ1}%
\end{equation}%
\begin{equation}
\exists\rho\in P_{n},\forall t\in\mathbf{R},\Phi^{\rho}(t,\mu)=\mu,
\label{equ2}%
\end{equation}%
\begin{equation}
\forall\rho\in P_{n},\forall t\in\mathbf{R},\Phi^{\rho}(t,\mu)=\mu,
\label{equ2_}%
\end{equation}%
\begin{equation}
\exists\rho\in P_{n},Or_{\rho}(\mu)=\{\mu\}, \label{equ3}%
\end{equation}%
\begin{equation}
\forall\rho\in P_{n},Or_{\rho}(\mu)=\{\mu\}, \label{equ3_}%
\end{equation}%
\begin{equation}
\mu\in NC_{1}\cap...\cap NC_{n}. \label{equ3__}%
\end{equation}

\end{theorem}

\begin{proof}
(\ref{equ1})$\Longrightarrow$(\ref{equ2}) We take $\rho\in P_{n}$ in the
following way%
\[
\rho(t)=(1,...,1)\cdot\chi_{\{t_{0}\}}(t)\oplus...\oplus(1,...,1)\cdot
\chi_{\{t_{k}\}}(t)\oplus...
\]
with $(t_{k})\in Seq.$ For the sequence%
\[
\forall k\in\mathbf{N},\alpha^{k}=(1,...,1)
\]
from $\Pi_{n}$ we can prove by induction on $k$ that%
\begin{equation}
\forall k\in\mathbf{N},\Phi^{\alpha^{0}...\alpha^{k}}(\mu)=\mu\label{equ6}%
\end{equation}
wherefrom%
\begin{equation}
\Phi^{\rho}(t,\mu)=\mu\cdot\chi_{(-\infty,t_{0})}(t)\oplus\mu\cdot
\chi_{\lbrack t_{0},t_{1})}(t)\oplus...\oplus\mu\cdot\chi_{\lbrack
t_{k},t_{k+1})}(t)\oplus...=\mu\label{equ7}%
\end{equation}

(\ref{equ2})$\Longrightarrow$(\ref{equ1}) From (\ref{equ2}) we have the
existence of $\alpha\in\Pi_{n}$ and $(t_{k})\in Seq$ with the property that
(\ref{equ7}) is true, thus (\ref{equ6}) is true. We denote%
\[
I_{0}=\{i|i\in\{1,...,n\},\alpha_{i}^{0}=1\},
\]%
\[
I_{1}=\{i|i\in\{1,...,n\},\alpha_{i}^{1}=1\},
\]%
\[
...
\]%
\[
I_{k}=\{i|i\in\{1,...,n\},\alpha_{i}^{k}=1\},
\]%
\[
...
\]
and we have from (\ref{equ6}):

$\forall i\in\{1,..,n\},$%
\[
\Phi_{i}^{\alpha^{0}}(\mu)=\left\{
\begin{array}
[c]{c}%
\Phi_{i}(\mu),i\in I_{0}\\
\mu_{i},i\in\{1,...,n\}\setminus I_{0}%
\end{array}
\right.  =\mu_{i};
\]

$\forall i\in\{1,..,n\},\Phi_{i}^{\alpha^{0}\alpha^{1}}(\mu)=\Phi_{i}%
^{\alpha^{1}}(\Phi^{\alpha^{0}}(\mu))=$%
\[
=\Phi_{i}^{\alpha^{1}}(\mu)=\left\{
\begin{array}
[c]{c}%
\Phi_{i}(\mu),i\in I_{1}\\
\mu_{i},i\in\{1,...,n\}\setminus I_{1}%
\end{array}
\right.  =\mu_{i};
\]%
\[
...
\]

$\forall i\in\{1,..,n\},\Phi_{i}^{\alpha^{0}\alpha^{1}...\alpha^{k}}(\mu
)=\Phi_{i}^{\alpha^{k}}(\Phi^{\alpha^{0}...\alpha^{k-1}}(\mu))=$%
\[
=\Phi_{i}^{\alpha^{k}}(\mu)=\left\{
\begin{array}
[c]{c}%
\Phi_{i}(\mu),i\in I_{k}\\
\mu_{i},i\in\{1,...,n\}\setminus I_{k}%
\end{array}
\right.  =\mu_{i};
\]%
\[
...
\]
with the conclusion that%
\[
\forall k\in\mathbf{N},\forall i\in I_{0}\cup I_{1}\cup...\cup I_{k},\Phi
_{i}(\mu)=\mu_{i}.
\]
Some $k^{\prime}\in\mathbf{N}$ exists with the property that%
\[
I_{0}\cup I_{1}\cup...\cup I_{k^{\prime}}=\{1,...,n\},
\]
thus (\ref{equ1}) is true.

(\ref{equ1})$\Longrightarrow$(\ref{equ2_}) Let be%
\begin{equation}
\rho(t)=\alpha^{0}\cdot\chi_{\{t_{0}\}}(t)\oplus...\oplus\alpha^{k}\cdot
\chi_{\{t_{k}\}}(t)\oplus... \label{equ8}%
\end{equation}
with $\alpha^{0},...,\alpha^{k},...\in\Pi_{n}$ and $(t_{k})\in Seq$ arbitrary.
It is proved by induction on $k$ the validity of (\ref{equ6}) and this implies
the truth of (\ref{equ7}).

(\ref{equ2_})$\Longrightarrow$(\ref{equ1}) This is true because (\ref{equ2_}%
)$\Longrightarrow$(\ref{equ2}) and (\ref{equ2})$\Longrightarrow$(\ref{equ1})
are true.

(\ref{equ2})$\Longleftrightarrow$(\ref{equ3}) and (\ref{equ2_}%
)$\Longleftrightarrow$(\ref{equ3_}) are obvious.

(\ref{equ1})$\Longleftrightarrow$(\ref{equ3__}) $\Phi(\mu)=\mu
\Longleftrightarrow\Phi_{1}(\mu)=\mu_{1}$ and...and $\Phi_{n}(\mu)=\mu
_{n}\Longleftrightarrow\mu\in NC_{1}$ and...and $\mu\in NC_{n}%
\Longleftrightarrow\mu\in NC_{1}\cap...\cap NC_{n}.$
\end{proof}

\begin{definition}
If $\Phi(\mu)=\mu,$ then $\forall\rho\in P_{n},$ the orbit $\Phi^{\rho}%
(t,\mu)=\mu$ is called \textbf{rest position}.
\end{definition}

\section{Fixed points vs. final values of the orbits}

\begin{theorem}
\label{The8}(\cite{bib7}, Theorem 49) The following fixed point property is
true%
\[
\forall\mu\in\mathbf{B}^{n},\forall\mu^{\prime}\in\mathbf{B}^{n},\forall
\rho\in P_{n},\underset{t\rightarrow\infty}{\lim}\Phi^{\rho}(t,\mu
)=\mu^{\prime}\Longrightarrow\Phi(\mu^{\prime})=\mu^{\prime}.
\]

\end{theorem}

\begin{proof}
Let $\mu\in\mathbf{B}^{n},\mu^{\prime}\in\mathbf{B}^{n},\rho\in P_{n}$ be
arbitrary and fixed. Some $t^{\prime}\in\mathbf{R}$ exists such that $\forall
t\geq t^{\prime},$%
\[
\mu^{\prime}=\Phi^{\rho}(t,\mu)\overset{\text{Remark \ref{Rem32}}}{=}%
\Phi^{\rho\cdot\chi_{(t^{\prime},\infty)}}(t,\mu^{\prime})
\]
and from Theorem \ref{The38}, (\ref{equ2})$\Longrightarrow$(\ref{equ1}) we
have $\Phi(\mu^{\prime})=\mu^{\prime}.$
\end{proof}

\begin{remark}
Theorem \ref{The8} shows that the final values of the states of a system are
fixed points of $\Phi$.
\end{remark}

\begin{theorem}
\label{The9}(\cite{bib7}, Theorem 50) We have $\forall\mu\in\mathbf{B}%
^{n},\forall\mu^{\prime}\in\mathbf{B}^{n},\forall\rho\in P_{n},$%
\[
(\Phi(\mu^{\prime})=\mu^{\prime}\;and\;\exists t^{\prime}\in\mathbf{R}%
,\Phi^{\rho}(t^{\prime},\mu)=\mu^{\prime})\Longrightarrow\forall t\geq
t^{\prime},\Phi^{\rho}(t,\mu)=\mu^{\prime}.
\]

\end{theorem}

\begin{proof}
For arbitrary $\mu\in\mathbf{B}^{n},\mu^{\prime}\in\mathbf{B}^{n},\rho\in
P_{n}$ we suppose that $\Phi(\mu^{\prime})=\mu^{\prime}\;$and$\;\Phi^{\rho
}(t^{\prime},\mu)=\mu^{\prime}.$ We have $\forall t\geq t^{\prime},$%
\[
\Phi^{\rho}(t,\mu)\overset{\text{Remark \ref{Rem32}}}{=}\Phi^{\rho\cdot
\chi_{(t^{\prime},\infty)}}(t,\mu^{\prime})\overset{\text{Theorem \ref{The38},
(\ref{equ1})}\Longrightarrow\text{(\ref{equ2_})}}{=}\mu^{\prime}.
\]

\end{proof}

\begin{remark}
As resulting from Theorem \ref{The9}, the accessible fixed points are final
values of the states of the systems.

The properties of the fixed points that are expressed by Theorems \ref{The38},
\ref{The8}, \ref{The9} give a better understanding of Example \ref{Exa34}.
\end{remark}

\section{Transitivity}

\begin{definition}
\label{Def108}The system $\Xi_{\Phi}$ (or $\Phi$) is \textbf{transitive}
(\cite{bib4}, page 22; \cite{bib3}, page 3), or \textbf{minimal} (\cite{bib4},
page 23) if one of the following non-equivalent properties holds true:%
\begin{equation}
\forall\mu\in\mathbf{B}^{n},\forall\mu^{\prime}\in\mathbf{B}^{n},\exists
\rho\in P_{n},\exists t\in\mathbf{R},\Phi^{\rho}(t,\mu)=\mu^{\prime},
\label{tt1}%
\end{equation}%
\begin{equation}
\forall\mu\in\mathbf{B}^{n},\forall\mu^{\prime}\in\mathbf{B}^{n},\forall
\rho\in P_{n},\exists t\in\mathbf{R},\Phi^{\rho}(t,\mu)=\mu^{\prime}.
\label{tt2}%
\end{equation}

\end{definition}

\begin{remark}
The property of transitivity may be considered one of surjectivity or one of accessibility.

If $\Phi$ is transitive, then it has no fixed points.
\end{remark}

\begin{example}
The property (\ref{tt1}) of transitivity is exemplified in Figure
\ref{minimality1}
\begin{figure}
[ptb]
\begin{center}
\fbox{\includegraphics[
height=0.9202in,
width=1.3422in
]%
{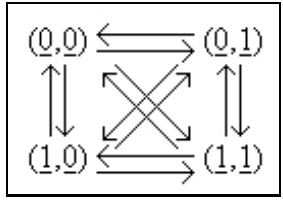}%
}\caption{Transitivity}%
\label{minimality1}%
\end{center}
\end{figure}
and the property (\ref{tt2}) of transitivity is exemplified in Figure
\ref{minimality2}.
\begin{figure}
[ptb]
\begin{center}
\fbox{\includegraphics[
height=0.9115in,
width=1.3327in
]%
{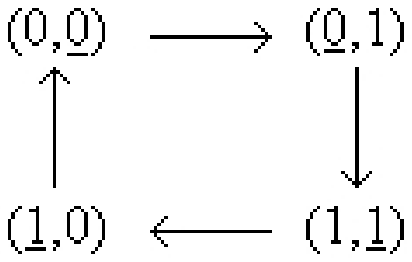}%
}\caption{Transitivity}%
\label{minimality2}%
\end{center}
\end{figure}

\end{example}

\section{The equivalence of the dynamical systems}

\begin{notation}
Let $h:\mathbf{B}^{n}\rightarrow\mathbf{B}^{n}$ and $x:\mathbf{R}%
\rightarrow\mathbf{B}^{n}$ be some functions. We denote by $h(x):\mathbf{R}%
\rightarrow\mathbf{B}^{n}$ the function%
\[
\forall t\in\mathbf{R},h(x)(t)=h(x(t)).
\]

\end{notation}

\begin{remark}
If $h:\mathbf{B}^{n}\rightarrow\mathbf{B}^{n}$ and $x\in S^{(n)}$ is expressed
by%
\[
x(t)=x(-\infty+0)\cdot\chi_{(-\infty,t_{0})}(t)\oplus x(t_{0})\cdot
\chi_{\lbrack t_{0},t_{1})}(t)\oplus...\oplus x(t_{k})\cdot\chi_{\lbrack
t_{k},t_{k+1})}(t)\oplus...
\]
then%
\[
h(x)(t)=h(x(-\infty+0))\cdot\chi_{(-\infty,t_{0})}(t)\oplus h(x(t_{0}%
))\cdot\chi_{\lbrack t_{0},t_{1})}(t)\oplus...
\]%
\[
...\oplus h(x(t_{k}))\cdot\chi_{\lbrack t_{k},t_{k+1})}(t)\oplus...
\]

\end{remark}

\begin{notation}
For $h:\mathbf{B}^{n}\rightarrow\mathbf{B}^{n}$ and $\alpha=\alpha
^{0},...,\alpha^{k},...\in\mathbf{B}^{n},$ we denote by $\widehat{h}(\alpha)$
the sequence $h(\alpha^{0}),...,h(\alpha^{k}),...\in\mathbf{B}^{n}.$
\end{notation}

\begin{notation}
Let be $k\geq2$ arbitrary and we denote for $\mu^{1},...,\mu^{k}\in
\mathbf{B}^{n},$%
\[
\mu^{1}\cup...\cup\mu^{k}=(\mu_{1}^{1}\cup...\cup\mu_{1}^{k},...,\mu_{n}%
^{1}\cup...\cup\mu_{n}^{k}).
\]

\end{notation}

\begin{notation}
We denote by $\Omega_{n}$ the set of the functions $h:\mathbf{B}%
^{n}\rightarrow\mathbf{B}^{n}$ that fulfill

i) $h$ is bijective;

ii) $h(0,...,0)=(0,...,0),\;h(1,...,1)=(1,...,1);$

iii) $\forall k\geq2,\forall\mu^{1}\in\mathbf{B}^{n},...,\forall\mu^{k}%
\in\mathbf{B}^{n},$%
\[
\mu^{1}\cup...\cup\mu^{k}=(1,...,1)\Longleftrightarrow h(\mu^{1})\cup...\cup
h(\mu^{k})=(1,...,1).
\]

\end{notation}

\begin{theorem}
\label{The59}a) $\Omega_{n}$ is group relative to the composition $^{\prime
}\circ^{\prime}$ of the functions;

b) $\forall h\in\Omega_{n},\forall\alpha\in\Pi_{n},\widehat{h}(\alpha)\in
\Pi_{n};$

c) $\forall h\in\Omega_{n},\forall\rho\in P_{n},h(\rho)\in P_{n}.$
\end{theorem}

\begin{proof}
a) The fact that $1_{\mathbf{B}^{n}}\in\Omega_{n},$ $\forall h\in\Omega
_{n},\forall h^{\prime}\in\Omega_{n},h\circ h^{\prime}\in\Omega_{n}$ and
$\forall h\in\Omega_{n},h^{-1}\in\Omega_{n}$ is obvious.

b) Let $h\in\Omega_{n}$ and $\alpha=\alpha^{0},...,\alpha^{k},...\in
\mathbf{B}^{n}$ be arbitrary. We denote for $p\geq1$%
\[
\{\mu^{1},...,\mu^{p}\}=\{\mu|\mu\in\mathbf{B}^{n},\{k|k\in\mathbf{N}%
,\alpha^{k}=\mu\}\text{ is infinite}\}
\]
and we remark that%
\[
\alpha\in\Pi_{n}\Longleftrightarrow\mu^{1},...,\mu^{p},\mu^{1},...,\mu^{p}%
,\mu^{1},...\in\Pi_{n}\Longleftrightarrow
\]%
\[
\Longleftrightarrow\left\{
\begin{array}
[c]{c}%
\mu^{1}=(1,...,1),p=1\\
\mu^{1}\cup...\cup\mu^{p}=(1,...,1),p\geq2
\end{array}
\right.  ,
\]%
\[
\widehat{h}(\alpha)\in\Pi_{n}\Longleftrightarrow h(\mu^{1}),...,h(\mu
^{p}),h(\mu^{1}),...,h(\mu^{p}),h(\mu^{1}),...\in\Pi_{n}\Longleftrightarrow
\]%
\[
\Longleftrightarrow\left\{
\begin{array}
[c]{c}%
h(\mu^{1})=(1,...,1),p=1\\
h(\mu^{1})\cup...\cup h(\mu^{p})=(1,...,1),p\geq2
\end{array}
\right.  .
\]

Case $p=1,$%
\[
\alpha\in\Pi_{n}\Longrightarrow\mu^{1}=(1,...,1)\Longrightarrow h(\mu
^{1})=(1,...,1)\Longrightarrow\widehat{h}(\alpha)\in\Pi_{n}.
\]

Case $p\geq2,$%
\[
\alpha\in\Pi_{n}\Longrightarrow\mu^{1}\cup...\cup\mu^{p}%
=(1,...,1)\Longrightarrow h(\mu^{1})\cup...\cup h(\mu^{p}%
)=(1,...,1)\Longrightarrow
\]%
\[
\Longrightarrow\widehat{h}(\alpha)\in\Pi_{n}.
\]

c) Let us take arbitrarily some $h\in\Omega_{n}$ and a function $\rho\in
P_{n},$%
\[
\rho(t)=\alpha^{0}\cdot\chi_{\{t_{0}\}}(t)\oplus...\oplus\alpha^{k}\cdot
\chi_{\{t_{k}\}}(t)\oplus...
\]
where $\alpha\in\Pi_{n}$ and $(t_{k})\in Seq.$ We have%
\[
h(\rho)(t)=h(\rho(t))=
\]%
\[
=h((0,...,0)\cdot\chi_{(-\infty,t_{0})}(t)\oplus\alpha^{0}\cdot\chi
_{\{t_{0}\}}(t)\oplus(0,...,0)\cdot\chi_{(t_{0},t_{1})}(t)\oplus...
\]%
\[
...\oplus\alpha^{k}\cdot\chi_{\{t_{k}\}}(t)\oplus(0,...,0)\cdot\chi
_{(t_{k},t_{k+1})}(t)\oplus...)
\]%
\[
=h(0,...,0)\cdot\chi_{(-\infty,t_{0})}(t)\oplus h(\alpha^{0})\cdot
\chi_{\{t_{0}\}}(t)\oplus h(0,...,0)\cdot\chi_{(t_{0},t_{1})}(t)\oplus...
\]%
\[
...\oplus h(\alpha^{k})\cdot\chi_{\{t_{k}\}}(t)\oplus h(0,...,0)\cdot
\chi_{(t_{k},t_{k+1})}(t)\oplus...
\]%
\[
=h(\alpha^{0})\cdot\chi_{\{t_{0}\}}(t)\oplus...\oplus h(\alpha^{k})\cdot
\chi_{\{t_{k}\}}(t)\oplus...
\]
Because $\widehat{h}(\alpha)\in\Pi_{n},$ taking into account b), we conclude
that $h(\rho)\in P_{n}.$
\end{proof}

\begin{theorem}
\label{The60}Let be the generator functions $\Phi,\Psi:\mathbf{B}%
^{n}\rightarrow\mathbf{B}^{n}$ of the systems $\Xi_{\Phi},\Xi_{\Psi}$ and the
bijections $h:\mathbf{B}^{n}\rightarrow\mathbf{B}^{n},h^{\prime}\in\Omega
_{n}.$ The following statements are equivalent:

a) $\forall\nu\in\mathbf{B}^{n},$ the diagrams%
\[%
\begin{array}
[c]{ccc}%
\mathbf{B}^{n} & \overset{\Phi^{\nu}}{\rightarrow} & \mathbf{B}^{n}\\
h\downarrow\; &  & \;\downarrow h\\
\mathbf{B}^{n} & \overset{\Psi^{h^{\prime}(\nu)}}{\rightarrow} &
\mathbf{B}^{n}%
\end{array}
\]
are commutative;

b) $\forall\mu\in\mathbf{B}^{n},\forall\alpha\in\Pi_{n},\forall k\in
\mathbf{N}_{\_},$%
\[
h(\widehat{\Phi}^{\alpha}(k,\mu))=\widehat{\Psi}^{\widehat{h^{\prime}}%
(\alpha)}(k,h(\mu));
\]

c) $\forall\mu\in\mathbf{B}^{n},\forall\rho\in P_{n},\forall t\in\mathbf{R},$%
\begin{equation}
h(\Phi^{\rho}(t,\mu))=\Psi^{h^{\prime}(\rho)}(t,h(\mu)). \label{eds1}%
\end{equation}

\end{theorem}

\begin{proof}
a)$\Longrightarrow$b) It is sufficient to prove that $\forall\mu\in
\mathbf{B}^{n},\forall\alpha\in\Pi_{n},\forall k\in\mathbf{N},$%
\begin{equation}
h(\Phi^{\alpha^{0}...\alpha^{k}}(\mu))=\Psi^{h^{\prime}(\alpha^{0}%
)...h^{\prime}(\alpha^{k})}(h(\mu)) \label{lem1_}%
\end{equation}
since this is equivalent with b).

We fix arbitrarily some $\mu$ and some $\alpha$ and we use the induction on
$k$. For $k=0$ the statement is proved, thus we suppose that it is true for
$k$ and we prove it for $k+1$:%
\[
h(\Phi^{\alpha^{0}...\alpha^{k}\alpha^{k+1}}(\mu))=h(\Phi^{\alpha^{k+1}}%
(\Phi^{\alpha^{0}...\alpha^{k}}(\mu)))=\Psi^{h^{\prime}(\alpha^{k+1})}%
(h(\Phi^{\alpha^{0}...\alpha^{k}}(\mu)))=
\]%
\[
=\Psi^{h^{\prime}(\alpha^{k+1})}(\Psi^{h^{\prime}(\alpha^{0})...h^{\prime
}(\alpha^{k})}(h(\mu)))=\Psi^{h^{\prime}(\alpha^{0})...h^{\prime}(\alpha
^{k})h^{\prime}(\alpha^{k+1})}(h(\mu)).
\]

b)$\Longrightarrow$c) For arbitrary $\mu\in\mathbf{B}^{n}$ and $\rho\in
P_{n},$%
\[
\rho(t)=\rho(t_{0})\cdot\chi_{\{t_{0}\}}(t)\oplus...\oplus\rho(t_{k})\cdot
\chi_{\{t_{k}\}}(t)\oplus...
\]
$(t_{k})\in Seq,\rho(t_{0}),...,\rho(t_{k}),...\in\Pi_{n}$ we have that%
\begin{equation}
h^{\prime}(\rho)(t)=h^{\prime}(\rho(t))=h^{\prime}(\rho(t_{0}))\cdot
\chi_{\{t_{0}\}}(t)\oplus...\oplus h^{\prime}(\rho(t_{k}))\cdot\chi
_{\{t_{k}\}}(t)\oplus... \label{lem1}%
\end{equation}
is an element of $P_{n}$ (see Theorem \ref{The59} c)) and%
\[
h(\Phi^{\rho}(t,\mu))=h(\mu\cdot\chi_{(-\infty,t_{0})}(t)\oplus\Phi
^{\rho(t_{0})}(\mu)\cdot\chi_{\lbrack t_{0},t_{1})}(t)\oplus...
\]%
\[
...\oplus\Phi^{\rho(t_{0})...\rho(t_{k})}(\mu)\cdot\chi_{\lbrack t_{k}%
,t_{k+1})}(t)\oplus...)=
\]%
\[
=h(\mu)\cdot\chi_{(-\infty,t_{0})}(t)\oplus h(\Phi^{\rho(t_{0})}(\mu
))\cdot\chi_{\lbrack t_{0},t_{1})}(t)\oplus...
\]%
\[
...\oplus h(\Phi^{\rho(t_{0})...\rho(t_{k})}(\mu))\cdot\chi_{\lbrack
t_{k},t_{k+1})}(t)\oplus...=
\]%
\[
\overset{(\ref{lem1_})}{=}h(\mu)\cdot\chi_{(-\infty,t_{0})}(t)\oplus
\Psi^{h^{\prime}(\rho(t_{0}))}(h(\mu))\cdot\chi_{\lbrack t_{0},t_{1}%
)}(t)\oplus...
\]%
\[
...\oplus\Psi^{h^{\prime}(\rho(t_{0}))...h^{\prime}(\rho(t_{k}))}(h(\mu
))\cdot\chi_{\lbrack t_{k},t_{k+1})}(t)\oplus...\overset{(\ref{lem1})}{=}%
\Psi^{h^{\prime}(\rho)}(t,h(\mu)).
\]

c)$\Longrightarrow$a) Let $\nu,\mu\in\mathbf{B}^{n}$ be arbitrary and fixed
and we consider $\rho\in P_{n},$%
\[
\rho(t)=\nu\cdot\chi_{\{t_{0}\}}(t)\oplus\rho(t_{1})\cdot\chi_{\{t_{1}%
\}}(t)\oplus...\oplus\rho(t_{k})\cdot\chi_{\{t_{k}\}}(t)\oplus...
\]
with $(t_{k})\in Seq$ fixed too. We have%
\[
h(\Phi^{\rho}(t,\mu))=h(\mu\cdot\chi_{(-\infty,t_{0})}(t)\oplus\Phi^{\nu}%
(\mu)\cdot\chi_{\lbrack t_{0},t_{1})}(t)\oplus\Phi^{\nu\rho(t_{1})}(\mu
)\cdot\chi_{\lbrack t_{1},t_{2})}(t)\oplus...)=
\]%
\[
=h(\mu)\cdot\chi_{(-\infty,t_{0})}(t)\oplus h(\Phi^{\nu}(\mu))\cdot
\chi_{\lbrack t_{0},t_{1})}(t)\oplus h(\Phi^{\nu\rho(t_{1})}(\mu))\cdot
\chi_{\lbrack t_{1},t_{2})}(t)\oplus...
\]
But%
\[
h^{\prime}(\rho)(t)=h^{\prime}(\rho(t))=h^{\prime}(\nu)\cdot\chi_{\{t_{0}%
\}}(t)\oplus h^{\prime}(\rho(t_{1}))\cdot\chi_{\{t_{1}\}}(t)\oplus...,
\]%
\[
\Psi^{h^{\prime}(\rho)}(t,h(\mu))=
\]%
\[
=h(\mu)\cdot\chi_{(-\infty,t_{0})}(t)\oplus\Psi^{h^{\prime}(\nu)}\cdot
\chi_{\lbrack t_{0},t_{1})}(t)\oplus\Psi^{h^{\prime}(\nu)h^{\prime}(\rho
(t_{1}))}\cdot\chi_{\lbrack t_{1},t_{2})}(t)\oplus...
\]
and from (\ref{eds1}), for $t\in\lbrack t_{0},t_{1}),$ we obtain%
\[
h(\Phi^{\nu}(\mu))=\Psi^{h^{\prime}(\nu)}(h(\mu)).
\]

\end{proof}

\begin{definition}
\label{Def160}We consider the generator functions $\Phi,\Psi:\mathbf{B}%
^{n}\rightarrow\mathbf{B}^{n}$ and the universal asynchronous systems
$\Xi_{\Phi},$ $\Xi_{\Psi}$. If two bijections $h:\mathbf{B}^{n}\rightarrow
\mathbf{B}^{n},h^{\prime}\in\Omega_{n}$ exist such that one of the equivalent
properties a), b), c) from Theorem \ref{The60} is satisfied, then $\Xi_{\Phi
},\Xi_{\Psi}$ are called \textbf{equivalent} (\cite{bib4}, page 35;
\cite{bib6}, page 102; \cite{bib1}, page 40; \cite{bib2}, page 32;
\cite{bib5}, page 6) and $\Phi,\Psi$ are called \textbf{conjugated}. In this
case we denote $\Phi\overset{(h,h^{\prime})}{\rightarrow}\Psi.$
\end{definition}

\begin{definition}
\label{Def161}We fix $\Phi.$ The fact that $\Psi\neq\Phi$ exists such that the
previous property holds, makes us say that $\Phi$ is \textbf{structurally
stable} (Peixoto \cite{bib6}, page 121). $\Psi$ is called an
\textbf{admissible} (or \textbf{allowable})\textbf{\ perturbation of} $\Phi$.
\end{definition}

\begin{remark}
The equivalence of the universal regular autonomous asynchronous systems is
indeed an equivalence and it should be understood as a change of coordinates.
Thus $\Phi$ and $\Psi$ are indistinguishable.
\end{remark}

\begin{example}
$\Phi,\Psi:\mathbf{B}^{2}\rightarrow\mathbf{B}^{2}$ are given by, see Figure
\ref{echiv1}%
\begin{figure}
[ptb]
\begin{center}
\fbox{\includegraphics[
height=1.4313in,
width=3.1981in
]%
{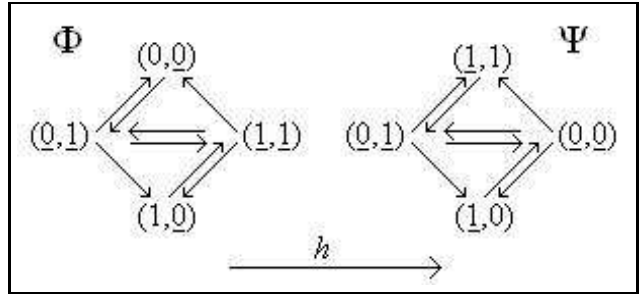}%
}\caption{Equivalent systems}%
\label{echiv1}%
\end{center}
\end{figure}
\[
\forall(\mu_{1},\mu_{2})\in\mathbf{B}^{2},\Phi(\mu_{1},\mu_{2})=(\mu_{1}%
\oplus\mu_{2},\overline{\mu_{2}}),
\]%
\[
\forall(\mu_{1},\mu_{2})\in\mathbf{B}^{2},\Psi(\mu_{1},\mu_{2})=(\overline
{\mu_{1}},\overline{\mu_{1}}\overline{\mu_{2}}\cup\mu_{1}\mu_{2})
\]
and the bijection $h:\mathbf{B}^{2}\rightarrow\mathbf{B}^{2}$ is%
\[
\forall(\mu_{1},\mu_{2})\in\mathbf{B}^{2},h(\mu_{1},\mu_{2})=(\overline
{\mu_{2}},\overline{\mu_{1}}).
\]
The diagram%
\[%
\begin{array}
[c]{ccc}%
\mathbf{B}^{2} & \overset{\Phi^{\nu}}{\rightarrow} & \mathbf{B}^{2}\\
h\downarrow\; &  & \;\downarrow h\\
\mathbf{B}^{2} & \overset{\Psi^{\nu^{\prime}}}{\rightarrow} & \mathbf{B}^{2}%
\end{array}
\]
commutes for $\nu=\nu^{\prime}=(0,0)$ and for $\nu=\nu^{\prime}=(1,1)$ we have
the assignments%
\[%
\begin{array}
[c]{ccc}%
(0,0) & \overset{\Phi}{\rightarrow} & (0,1)\\
h\downarrow\; &  & \;\downarrow h\\
(1,1) & \overset{\Psi}{\rightarrow} & (0,1)
\end{array}
,\;%
\begin{array}
[c]{ccc}%
(0,1) & \overset{\Phi}{\rightarrow} & (1,0)\\
h\downarrow\; &  & \;\downarrow h\\
(0,1) & \overset{\Psi}{\rightarrow} & (1,0)
\end{array}
,\;%
\begin{array}
[c]{ccc}%
(1,0) & \overset{\Phi}{\rightarrow} & (1,1)\\
h\downarrow\; &  & \;\downarrow h\\
(1,0) & \overset{\Psi}{\rightarrow} & (0,0)
\end{array}
,\;%
\begin{array}
[c]{ccc}%
(1,1) & \overset{\Phi}{\rightarrow} & (0,0)\\
h\downarrow\; &  & \;\downarrow h\\
(0,0) & \overset{\Psi}{\rightarrow} & (1,1)
\end{array}
.
\]
We denote $\pi_{i}:\mathbf{B}^{2}\rightarrow\mathbf{B},\forall(\mu_{1},\mu
_{2})\in\mathbf{B}^{2},$%
\[
\pi_{i}(\mu_{1},\mu_{2})=\mu_{i},i=\overline{1,2}.
\]
For $\nu=(0,1),\nu^{\prime}=(1,0)$ we have%
\[%
\begin{array}
[c]{ccc}%
(0,0) & \overset{(\pi_{1},\Phi_{2})}{\rightarrow} & (0,1)\\
h\downarrow\; &  & \;\downarrow h\\
(1,1) & \overset{(\Psi_{1},\pi_{2})}{\rightarrow} & (0,1)
\end{array}
,%
\begin{array}
[c]{ccc}%
(0,1) & \overset{(\pi_{1},\Phi_{2})}{\rightarrow} & (0,0)\\
h\downarrow\; &  & \;\downarrow h\\
(0,1) & \overset{(\Psi_{1},\pi_{2})}{\rightarrow} & (1,1)
\end{array}
,%
\begin{array}
[c]{ccc}%
(1,0) & \overset{(\pi_{1},\Phi_{2})}{\rightarrow} & (1,1)\\
h\downarrow\; &  & \;\downarrow h\\
(1,0) & \overset{(\Psi_{1},\pi_{2})}{\rightarrow} & (0,0)
\end{array}
,%
\begin{array}
[c]{ccc}%
(1,1) & \overset{(\pi_{1},\Phi_{2})}{\rightarrow} & (1,0)\\
h\downarrow\; &  & \;\downarrow h\\
(0,0) & \overset{(\Psi_{1},\pi_{2})}{\rightarrow} & (1,0)
\end{array}
\]
and for $\nu=(1,0),\nu^{\prime}=(0,1)$ the assignments are%
\[%
\begin{array}
[c]{ccc}%
(0,0) & \overset{(\Phi_{1},\pi_{2})}{\rightarrow} & (0,0)\\
h\downarrow\; &  & \;\downarrow h\\
(1,1) & \overset{(\pi_{1},\Psi_{2})}{\rightarrow} & (1,1)
\end{array}
,%
\begin{array}
[c]{ccc}%
(0,1) & \overset{(\Phi_{1},\pi_{2})}{\rightarrow} & (1,1)\\
h\downarrow\; &  & \;\downarrow h\\
(0,1) & \overset{(\pi_{1},\Psi_{2})}{\rightarrow} & (0,0)
\end{array}
,%
\begin{array}
[c]{ccc}%
(1,0) & \overset{(\Phi_{1},\pi_{2})}{\rightarrow} & (1,0)\\
h\downarrow\; &  & \;\downarrow h\\
(1,0) & \overset{(\pi_{1},\Psi_{2})}{\rightarrow} & (1,0)
\end{array}
,%
\begin{array}
[c]{ccc}%
(1,1) & \overset{(\Phi_{1},\pi_{2})}{\rightarrow} & (0,1)\\
h\downarrow\; &  & \;\downarrow h\\
(0,0) & \overset{(\pi_{1},\Psi_{2})}{\rightarrow} & (0,1)
\end{array}
\]
respectively. $\Phi$ and $\Psi$ are conjugated.
\end{example}

\begin{example}
The functions $h,h^{\prime}:\mathbf{B}^{2}\rightarrow\mathbf{B}^{2}$ are given
in the following table%
\[%
\begin{array}
[c]{ccc}%
(\mu_{1},\mu_{2}) & h(\mu_{1},\mu_{2}) & h^{\prime}(\mu_{1},\mu_{2})\\
(0,0) & (0,1) & (0,0)\\
(0,1) & (1,1) & (1,0)\\
(1,0) & (0,0) & (0,1)\\
(1,1) & (1,0) & (1,1)
\end{array}
\]
and the state portraits of the two systems are given in Figure \ref{echiv2}.
$\Xi_{\Phi}$ and $\Xi_{\Psi}$ are equivalent.%
\begin{figure}
[ptb]
\begin{center}
\fbox{\includegraphics[
height=1.3612in,
width=3.1981in
]%
{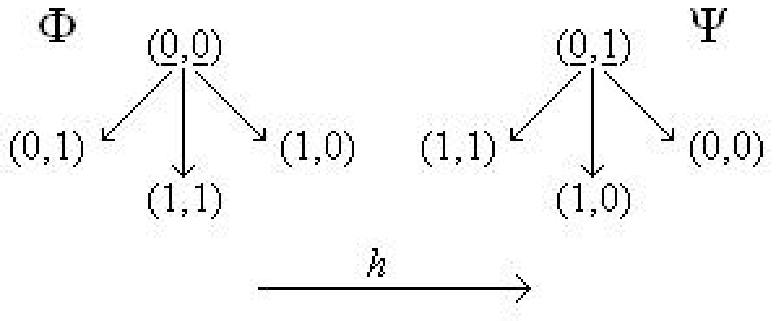}%
}\caption{Equivalent systems}%
\label{echiv2}%
\end{center}
\end{figure}

\end{example}

\begin{theorem}
If $\Phi$ and $\Psi$ are conjugated, then the following possibilities exist:

a) $\Phi=\Psi=1_{\mathbf{B}^{n}};$

b) $\Phi\neq1_{\mathbf{B}^{n}}$ and $\Psi\neq1_{\mathbf{B}^{n}}.$
\end{theorem}

\begin{proof}
We presume that $\Phi\overset{(h,h^{\prime})}{\rightarrow}\Psi.$ In the
equation%
\[
\forall\nu\in\mathbf{B}^{n},\forall\mu\in\mathbf{B}^{n},h(\Phi^{\nu}%
(\mu))=\Psi^{h^{\prime}(\nu)}(h(\mu))
\]
we put $\Psi=1_{\mathbf{B}^{n}}$ and we have%
\[
\forall\nu\in\mathbf{B}^{n},\forall\mu\in\mathbf{B}^{n},h(\Phi^{\nu}%
(\mu))=h(\mu)
\]
thus $\forall\nu\in\mathbf{B}^{n},\Phi^{\nu}=1_{\mathbf{B}^{n}}$ and finally
$\Phi=1_{\mathbf{B}^{n}}.$
\end{proof}

\begin{theorem}
We suppose that $\Xi_{\Phi}$ and $\Xi_{\Psi}$ are equivalent and let be
$h,h^{\prime}$ such that $\Phi\overset{(h,h^{\prime})}{\rightarrow}\Psi$.

a) If $\mu$ is a fixed point of $\Phi,$ then $h(\mu)$ is a fixed point of
$\Psi.$

b) For any $\mu\in\mathbf{B}^{n}$ and any $\rho\in P_{n},$ if $\Phi^{\rho
}(t,\mu)$ is periodical with the period $T_{0}$, then $\Psi^{h^{\prime}(\rho
)}(t,h(\mu))$ is periodical with the period $T_{0}$.

c) If $\Xi_{\Phi}$ is transitive, then $\Xi_{\Psi}$ is transitive.
\end{theorem}

\begin{proof}
a) The commutativity of the diagram%
\[%
\begin{array}
[c]{ccc}%
\mathbf{B}^{n} & \overset{\Phi^{\nu}}{\rightarrow} & \mathbf{B}^{n}\\
h\downarrow\; &  & \;\downarrow h\\
\mathbf{B}^{n} & \overset{\Psi^{h^{\prime}(\nu)}}{\rightarrow} &
\mathbf{B}^{n}%
\end{array}
\]
for $\nu=(1,...,1)$ gives%
\[
h(\mu)=h(\Phi(\mu))=h(\Phi^{(1,...,1)}(\mu))=\Psi^{h^{\prime}(1,...,1)}%
(h(\mu))=
\]%
\[
=\Psi^{(1,...,1)}(h(\mu))=\Psi(h(\mu)).
\]

b) The hypothesis states that $\exists t^{\prime}\in\mathbf{R},\forall t\geq
t^{\prime},$%
\[
\Phi^{\rho}(t,\mu)=\Phi^{\rho}(t+T_{0},\mu)
\]
and in this situation%
\[
\Psi^{h^{\prime}(\rho)}(t,h(\mu))=h(\Phi^{\rho}(t,\mu))=h(\Phi^{\rho}%
(t+T_{0},\mu))=\Psi^{h^{\prime}(\rho)}(t+T_{0},h(\mu)).
\]

c) Let $\mu,\mu^{\prime}\in\mathbf{B}^{n}$ be arbitrary and fixed. The
hypothesis (\ref{tt1}) states that%
\[
\exists\rho\in P_{n},\exists t\in\mathbf{R},\Phi^{\rho}(t,h^{-1}(\mu
))=h^{-1}(\mu^{\prime}),
\]
wherefrom%
\[
\Psi^{h^{\prime}(\rho)}(t,\mu)=\Psi^{h^{\prime}(\rho)}(t,h(h^{-1}%
(\mu)))=h(\Phi^{\rho}(t,h^{-1}(\mu))=h(h^{-1}(\mu^{\prime}))=\mu^{\prime}.
\]
The situation with (\ref{tt2}) is similar.
\end{proof}

\section{Dynamic bifurcations}

\begin{remark}
Let be the generator function $\Phi:\mathbf{B}^{n}\times\mathbf{B}%
^{m}\rightarrow\mathbf{B}^{n},$ $\mathbf{B}^{n}\times\mathbf{B}^{m}\ni
(\mu,\lambda)\rightarrow\Phi(\mu,\lambda)\in\mathbf{B}^{n}$ that depends on
the parameter $\lambda\in\mathbf{B}^{m}$. Intuitively speaking (Ott,
\cite{bib3}, page 137) a dynamic bifurcation is a qualitative change in the
dynamic of the system $\Xi_{\Phi(\cdot,\lambda)}$ that occurs at the variation
of the parameter $\lambda$.
\end{remark}

\begin{definition}
If for any parameters $\lambda,\lambda^{\prime}\in\mathbf{B}^{m}$ the systems
$\Xi_{\Phi(\cdot,\lambda)}$ and $\Xi_{\Phi(\cdot,\lambda^{\prime})}$ are
equivalent, then $\Phi$ is called \textbf{structurally stable} (\cite{bib6},
page 117; \cite{bib2}, page 43; \cite{bib5}, page 9); the existence of
$\lambda,\lambda^{\prime}$ such that $\Xi_{\Phi(\cdot,\lambda)}$ and
$\Xi_{\Phi(\cdot,\lambda^{\prime})}$ are not equivalent is called a
\textbf{dynamic} \textbf{bifurcation} (\cite{bib1}, page 57; \cite{bib5}, page 9).

Equivalently, let us fix an arbitrary $\lambda\in\mathbf{B}^{m}.$ If
$\forall\lambda^{\prime}\in\mathbf{B}^{m}$, $\Phi(\cdot,\lambda^{\prime})$ is
an admissible perturbation of $\Phi(\cdot,\lambda)$ (Definition \ref{Def161}),
then $\Phi$ is said to be \textbf{structurally stable}, otherwise we say that
$\Phi$ has a \textbf{dynamic bifurcation}.
\end{definition}

\begin{remark}
If $\forall\lambda\in\mathbf{B}^{m},\forall\lambda^{\prime}\in\mathbf{B}^{m}$
the bijections $h:\mathbf{B}^{n}\rightarrow\mathbf{B}^{n},h^{\prime}\in
\Omega_{n}$ exist such that $\forall\nu\in\mathbf{B}^{n},$ the diagram%
\[%
\begin{array}
[c]{ccc}%
\;\;\mathbf{B}^{n} & \overset{\Phi^{\nu}(\cdot,\lambda)}{\rightarrow} &
\mathbf{B}^{n}\;\\
h\downarrow &  & \downarrow h\\
\;\;\mathbf{B}^{n} & \overset{\Phi^{h^{\prime}(\nu)}(\cdot,\lambda^{\prime}%
)}{\rightarrow} & \mathbf{B}^{n}\;
\end{array}
\]
commutes, then $\Phi$ is structurally stable, otherwise we have a dynamic bifurcation.
\end{remark}

\begin{example}
In Figure \ref{bifurcation3} ($n=2,m=1$),
\begin{figure}
[ptb]
\begin{center}
\fbox{\includegraphics[
height=1.1537in,
width=3.0588in
]%
{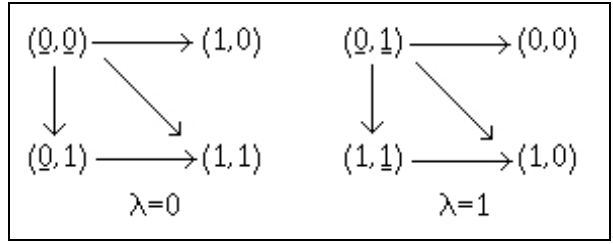}%
}\caption{Structural stability}%
\label{bifurcation3}%
\end{center}
\end{figure}
$\Phi$ is structurally stable and the bijections $h,h^{\prime}$ are defined
accordingly to the following table:%
\[%
\begin{array}
[c]{ccc}%
(\mu_{1},\mu_{2}) & h(\mu_{1},\mu_{2}) & h^{\prime}(\mu_{1},\mu_{2})\\
(0,0) & (0,1) & (0,0)\\
(0,1) & (1,1) & (1,0)\\
(1,0) & (0,0) & (0,1)\\
(1,1) & (1,0) & (1,1)
\end{array}
\]

\end{example}

\begin{example}
In Figure \ref{bifurcation4} ($n=2,m=1$),%
\begin{figure}
[ptb]
\begin{center}
\fbox{\includegraphics[
height=1.1078in,
width=2.9888in
]%
{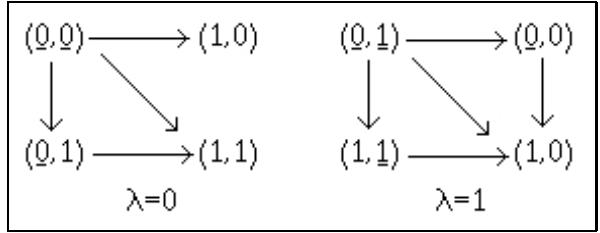}%
}\caption{Dynamic bifurcation}%
\label{bifurcation4}%
\end{center}
\end{figure}
$\Phi$ has a dynamic bifurcation.
\end{example}

\begin{definition}
The \textbf{bifurcation diagram} (\cite{bib1}, page 61) is a partition of the
set of systems $\{\Xi_{\Phi(\cdot,\lambda)}|\lambda\in\mathbf{B}^{m}\}$ in
classes of equivalence given by the equivalence of the systems, together with
representative state portraits for each class of equivalence.
\end{definition}

\begin{example}
Figure \ref{bifurcation4} is a bifurcation diagram.
\end{example}

\begin{definition}
The \textbf{bifurcation diagram} (\cite{bib3}, page 5) is the graph that gives
the position of the fixed points depending on a parameter, such that a
bifurcation exists.
\end{definition}

\begin{remark}
Such a(n informal) definition works for calling Figure \ref{bifurcation4} a
bifurcation diagram, since there fixed points exist. However for Figure
\ref{bifurcation5}
\begin{figure}
[ptb]
\begin{center}
\fbox{\includegraphics[
height=1.1087in,
width=2.9862in
]%
{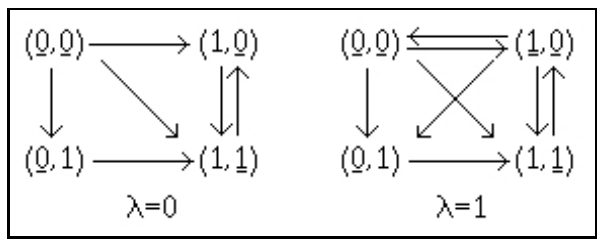}%
}\caption{Dynamic bifurcation}%
\label{bifurcation5}%
\end{center}
\end{figure}
this definition does not work, because a bifurcation exists there, but no
fixed points.
\end{remark}

\begin{definition}
Let be $\Phi,\Psi:\mathbf{B}^{n}\times\mathbf{B}^{m}\rightarrow\mathbf{B}^{n}%
$. The families of systems $(\Xi_{\Phi(\cdot,\lambda)})_{\lambda\in
\mathbf{B}^{m}}$ and $(\Xi_{\Psi(\cdot,\lambda)})_{\lambda\in\mathbf{B}^{m}}$
are called \textbf{equivalent} (\cite{bib5}, pages 7, 17) if there exists a
bijection $h^{\prime\prime}:\mathbf{B}^{m}\rightarrow\mathbf{B}^{m}$ such that
$\forall\lambda\in\mathbf{B}^{m},\Xi_{\Phi(\cdot,\lambda)}$ and $\Xi
_{\Psi(\cdot,h^{\prime\prime}(\lambda))}$ are equivalent in the sense of
Definition \ref{Def160}.
\end{definition}

\end{document}